\documentclass[USenglish]{lipics-v2016}

\newcommand{\sv}[1]{}
\newcommand{\lv}[1]{#1}

\usepackage{paralist}

\usepackage{amsthm,amsmath,amssymb,graphicx}
\usepackage{enumerate,verbatim}
\usepackage{enumerate}
\usepackage{xspace}

\usepackage[draft,author=]{fixme}
\fxsetup{theme=color}

\newcommand{\PPP}{\mathcal{P}}
\newcommand{\QQQ}{\mathcal{Q}}

\newcommand{\FFF}{\mathcal{F}}

\newcommand{\VVV}{\mathcal{V}}

\newcommand{\Nat}{\mathbb{N}}

\newcommand{\bigoh}{\mathcal{O}}

\theoremstyle{plain}
\newtheorem{THE}[theorem]{Theorem}
\newtheorem{PRO}[theorem]{Proposition}
\newtheorem{LEM}[theorem]{Lemma}

\newtheorem*{THE*}{Theorem}
\newtheorem*{fact*}{Fact}

\theoremstyle{definition}
\newtheorem{DEF}[theorem]{Definition}

\newtheorem*{DEF*}{Definition}
\newtheorem{proposition}[theorem]{Proposition}

\theoremstyle{remark}

\newtheorem{CLM}{Claim}

\newcommand{\var}{\mbox{var}}

\newcommand{\stw}{dependency treewidth}

\newcommand{\hy}{\hbox{-}\nobreak\hskip0pt}
\newcommand{\tw}{\textup{tw}}
\newcommand{\ptdt}[2]{F_{#2}(#1)}
\newcommand{\hubiewidth}{CD-width}

\begin{document}

\sv{\title{Small Resolution Proofs for QBF using Dependency Treewidth} }
\lv{\title{Small Resolution Proofs for QBF using Dependency Treewidth} }

\author{Eduard Eiben}
\author{Robert Ganian}
\author{Sebastian Ordyniak}
\affil{Algorithms and Complexity Group, TU Wien, Vienna, Austria\\
  \texttt{eduard.eiben@gmail.com, rganian@gmail.com, sordyniak@gmail.com}}

\maketitle

\begin{abstract}
In spite of the close connection between the evaluation of quantified
Boolean formulas (QBF) and propositional satisfiability (SAT), tools
and techniques which exploit structural properties of SAT instances
are known to fail for QBF. This is especially true for the structural
parameter treewidth, which has allowed the design of successful
algorithms for SAT but cannot be straightforwardly applied to QBF
since it does not take into account the interdependencies between
quantified variables.

In this work we introduce and develop dependency treewidth, a new structural
parameter based on treewidth which allows the efficient solution of
QBF instances. Dependency treewidth pushes the frontiers of tractability for QBF
by overcoming the limitations of previously introduced variants of
treewidth for QBF. 
We augment our results by developing algorithms for computing the
decompositions that are required to use the parameter.
\end{abstract}

\section{Introduction}
\label{sec:intro}

The problem of evaluating quantified Boolean formulas (QBF) is a generalization of the propositional satisfiability problem (SAT) which naturally captures a range of computational tasks in areas such as verification, planning, knowledge representation and automated reasoning~\cite{EglyEiterTW00,OtwellRemshagenTruemper04,Rintanen99,SabharwalAGHS06}. 
QBF is the archetypical \textsc{PSpace}-complete problem and is therefore believed to be computationally harder than \textsc{NP}-complete problems such as SAT~\cite{KleineBuningLettman99,Papadimitriou94,StockmeyerMeyer73}.

In spite of the close connection between QBF and SAT, many of the tools and techniques which work for SAT are known not to help for QBF, and dynamic programming based on the structural parameter treewidth~\cite{AtseriasOliva14,Szeider04b} is perhaps the most prominent example of this behavior. Treewidth is a highly-established measure of how ``treelike'' an instance is, 
and in the SAT setting it is known that $n$-variable instances of treewidth at most $k$ can be solved in time at most $f(k)\cdot n$~\cite{Szeider04b} for a computable function $f$. Algorithms with running time in this form (i.e., $f(k)\cdot n^{\bigoh(1)}$, where $k$ is the parameter and the degree of the polynomial of $n$ is independent of $k$) are called \emph{fixed-parameter algorithms}, and problems which admit such an algorithm (w.r.t.\ a certain parameter) belong to the class \emph{FPT}. Furthermore, in the SAT setting, treewidth allows us to do more than merely solve the instance: it is also possible to find a so-called \emph{resolution proof}~\cite{DavisPutnam60,Capellithesis16}. If the input was a non-instance, such a resolution proof contains additional information on ``what makes it unsatisfiable'' and hence can be more useful than outputting a mere \textbf{Reject} in practical settings.

In the QBF setting, the situation is considerably more complicated. It is known that QBF instances of bounded treewidth remain \textsc{PSpace}-complete~\cite{AtseriasOliva14}, and the intrinsic reason for this fact is that treewidth does not take into account the dependencies that arise between variables in QBF. So far, 
there have been 
several 
attempts at remedying this situation by introducing variants of treewidth which support fixed-parameter algorithms for QBF: \emph{prefix pathwidth} (along with \emph{prefix treewidth})~\cite{EibenGanianOrdyniak16} and \emph{respectful treewidth}~\cite{AtseriasOliva14}, along with two other parameters~\cite{AdlerWeyer12,ChenDalmau16} which originate from a different setting but can also be adapted to obtain fixed-parameter algorithms for QBF. We refer to Subsection~\ref{sub:comp} for a comparison of these parameters.
Aside from algorithms with runtime guarantees, it is worth noting 
that empirical connections between treewidth and QBF have also
been studied in the literature~\cite{PulinaT09,PulinaT10}.

In this work we introduce and develop dependency treewidth, a new structural parameter based on treewidth which supports fixed-parameter algorithms for QBF. Dependency treewidth pushes the frontiers of tractability for QBF by overcoming the limitations of both the previously introduced prefix and respectful variants. Compared to the former, this new parameter allows the computation of resolution proofs analogous to the case of classical treewidth for SAT instances. Prefix pathwidth relies on entirely different techniques to solve QBF and does not yield small resolution proofs. Moreover, the running time of the fixed-parameter algorithm which uses prefix pathwidth has a triple-exponential dependency on the parameter $k$, while dependency treewidth supports a considerably more efficient $\bigoh(3^{2k}nk)$-time algorithm for QBF.

Unlike respectful treewidth and its variants, which only take the basic dependencies between variables into account, dependency treewidth can be used in conjunction with the so-called \emph{dependency
  schemes} introduced by Samer and Szeider~\cite{SamerSzeider09a,SlivovskySzeider14}, see also the work of Biere and Lonsing~\cite{BiereLonsing10}. Dependency schemes allow an in-depth analysis of how the assignment of individual variables in a QBF depends on other
variables, and research in this direction has uncovered a large number
of distinct dependency schemes with varying complexities. The most basic dependency scheme is
called the \emph{trivial dependency scheme}~\cite{SamerSzeider09a},
which stipulates that each variable depends on all variables with
distinct quantification which precede it in the prefix. Respectful treewidth in fact coincides with dependency treewidth when the trivial dependency scheme is used, but more advanced dependency schemes allow us to efficiently solve instances which otherwise remain out of the reach of state-of-the-art techniques.

Crucially, all of the structural parameters mentioned above require a so-called \emph{decomposition} in order to solve QBF; computing these decompositions is typically an NP-hard problem. A large part of our technical contribution lies in developing algorithms to compute decompositions for dependency treewidth. Without such algorithms, it would not be possible to use the parameter unless a decomposition were supplied as part of the input (an unrealistic assumption in practical settings). 
It is worth noting that all of these algorithms can also be used to find respectful tree decompositions, where the question of finding suitable decompositions was left open~\cite{AtseriasOliva14}. 
We provide two algorithms for computing dependency tree decompositions, each suitable for use under different situations.

The article is structured as follows. After the preliminaries, we introduce the parameter and show how to use it to solve QBF (assuming a decomposition has already been computed). This section also contains an in-depth overview and comparison of previous work in the area. A separate section then introduces other equivalent characterizations of dependency treewidth. The last technical section contains our algorithms for finding dependency tree decompositions, after which we provide concluding notes and remarks.

\sv{\smallskip \noindent {\emph{Statements whose full proofs are located in the supplementary material are marked with $\star$.}}}

\section{Preliminaries}
\label{sec:prelims}
\newcommand{\guards}{\delta}

For $i\in \mathbb{N}$, we let $[i]$ denote the set $\{1,\dots,i\}$. We
refer to the book by Diestel~\cite{Diestel12} for standard graph
terminology. Given a graph $G$, we denote by $V(G)$ and $E(G)$ its vertex
and edge set, respectively. We use $ab$ as a shorthand for the edge
$\{a,b\}$. For $V'\subseteq V(G)$, the \emph{guards}
of $V'$ (denoted $\guards(V')$) are the vertices in $V(G)\setminus V'$
with at least one neighbor in $V'$.

\sv{
We refer to the standard textbooks~\cite{DowneyFellows13,FlumGrohe06} for an in-depth
overview of parameterized complexity theory. Here, we only recall that
a \emph{parameterized problem} $(Q,\kappa)$ is a 
\emph{problem} $Q \subseteq \Sigma^*$ together 
with a \emph{parameterization} $\kappa \colon \Sigma^* \to
\mathbb{N}$, where $\Sigma$ is a finite alphabet.  
A parameterized problem $(Q,\kappa)$ is \emph{fixed-parameter tractable (w.r.t.\   $\kappa$)}, 
in short \emph{FPT}, if there exists a decision algorithm for $Q$, 
a computable function $f$, 
and a polynomial function $p$, 
such that for all $x \in \Sigma^*$, the running time of the algorithm on $x$ 
is at most $f(\kappa(x)) \cdot p(|x|)$. Algorithms with this running
time are called \emph{fixed-parameter algorithms}.
}

\lv{
We refer to the standard textbooks~\cite{DowneyFellows13,FlumGrohe06} for an in-depth
overview of parameterized complexity theory. Here, we only recall that
a \emph{parameterized problem} $(Q,\kappa)$ is a 
\emph{problem} $Q \subseteq \Sigma^*$ together 
with a \emph{parameterization} $\kappa \colon \Sigma^* \to
\mathbb{N}$, where $\Sigma$ is a finite alphabet.  
A parameterized problem $(Q,\kappa)$ is \emph{fixed-parameter tractable (w.r.t.\   $\kappa$)}, 
in short \emph{FPT}, if there exists a decision algorithm for $Q$, 
a computable function $f \colon \mathbb{N} \to \mathbb{N}$, 
and a polynomial function $p \colon \mathbb{N} \to \mathbb{N}$, 
such that for all $x \in \Sigma^*$, the running time of the algorithm on $x$ 
is at most $f(\kappa(x)) \cdot p(|x|)$. Algorithms with this running
time are called \emph{fixed-parameter algorithms}.
}

\subsection{Quantified Boolean Formulas}

For a set of propositional variables $K$, a \emph{literal} is either a variable $x\in K$ or its negation $\bar{x}$. A \emph{clause} is a disjunction over literals. A \emph{propositional formula in conjunctive normal form} (i.e., a \emph{CNF formula}) is a conjunction over clauses. Given a CNF formula $\phi$, we denote the set of variables which occur in $\phi$ by $\var(\phi)$. For notational purposes, we will view a clause as a set of literals and a CNF formula as a set of clauses.

A \emph{quantified Boolean formula} is a tuple $(\phi, \tau)$ where $\phi$ is a CNF formula and $\tau$ is a sequence of quantified variables, denoted $\var(\tau)$, which satisfies $\var(\tau)\supseteq \var(\phi)$; then $\phi$ is called the \emph{matrix} and $\tau$ is called the \emph{prefix}. A QBF $(\phi, \tau)$ is true if the formula $\tau \phi$ is true. A \emph{quantifier block} is a maximal sequence of consecutive variables with the same quantifier. An \emph{assignment} is a mapping from (a subset of) the variables to $\{0,1\}$.

The \emph{primal graph} of a QBF $I=(\phi, \tau)$ is the graph $G_I$ defined as follows. The vertex set of $G_I$ consists of every variable which occurs in $\phi$, and $st$ is an edge in $G_I$ if there exists a clause in $\phi$ containing both $s$ and $t$.

\subsection{Dependency Posets for QBF}


\newcommand{\cov}{\lhd}
\newcommand{\incom}{\parallel}
\newcommand{\width}{\textup{width}}
\newcommand{\cl}[2]{\textnormal{cl}$_{#1}$(#2)}
Before proceeding, we define a few standard notions related to posets which will be used throughout the paper. A partially ordered set (\emph{poset}) $\VVV$ is a pair $(V,\leq^V)$ where $V$ is a set and $\leq^V$ is
a reflexive, antisymmetric, and transitive binary relation over $V$.
A \emph{chain} $W$ of
$\VVV$ is a subset of $V$ such that $x \leq^V y$ or $y \leq^V x$
for every $x,y \in W$. 
A \emph{chain partition} of $\VVV$ is a tuple $(W_1,\dotsc,W_k)$ such
that $\{W_1,\dotsc,W_k\}$ is a partition of $V$ and for every $i$ with
$1 \leq i \leq k$ the poset induced by $W_i$ is a chain of $\VVV$.
An \emph{anti-chain} $A$ of $\VVV$ is a subset
of $V$ such that for all $x,y \in A$ neither $x \leq^V\! y$ nor $y\leq^V\! x$.
The \emph{width} (or \emph{poset-width}) of a poset $\VVV$, denoted by $\width(\VVV)$, is the 
maximum cardinality of any anti-chain of $\VVV$.
A poset of width $1$ is called a \emph{linear order}. 
A \emph{linear extension} of a poset $\PPP = (P,\leq^P)$ is a relation $\preceq$ over $P$ such that $x\preceq y$ whenever $x\leq^P y$ and the poset $\PPP^*=(P, \preceq)$ is a linear order.
A subset $A$ of $V$
is \emph{downward-closed} if 
for every $a \in A$ it holds that $b\leq^V a\implies b\in A$. A \emph{reverse} of a poset is obtained by reversing each relation in the poset.
For brevity we will often write $\leq^V$ to refer to the poset
$\VVV:=(V,\leq^V)$. 

\lv{
\begin{PRO}[{\cite{frs03}}]\label{pro:comp-chain-part}
  Let $\VVV$ be a poset. Then in time $\bigoh(\width(\VVV)\cdot\Vert\VVV\Vert^2)$, it is
  possible to compute both $\width(\VVV)=w$ and a
  corresponding chain partition $(W_1,\dotsc,W_{w})$ of $\VVV$.
\end{PRO}
}

We use \emph{dependency posets} to provide a general and formal way of
speaking about the various \emph{dependency schemes} introduced for
QBF~\cite{SamerSzeider09a}. 
It is important to note that dependency schemes in general are too broad a notion for our purposes; for instance, it is known that some dependency schemes do not even give rise to sound resolution proof systems. Here we focus solely on so-called \emph{permutation dependency schemes}~\cite{SlivovskySzeider16}, which is a general class containing all commonly used dependency schemes that give rise to sound resolution proof systems.
This leads us to our definition of dependency posets, which allow us to capture all permutation dependency schemes. 


Given a QBF $I=(\phi, \tau)$, a
\emph{dependency poset} $\mathcal{V}=(\var(\phi), {\leq^{I}})$ of $I$ is a poset over
$\var(\phi)$ with the following properties:

\begin{compactenum}
\item for all $x,y \in \var(\phi)$, if $x\leq^I y$, then $x$ is before $y$ in the prefix, and 
\item given any linear extension $\preceq$ of $\VVV$, the QBF $I'=(\phi, \tau_\preceq)$, obtained by permutation of the prefix $\tau$ according the $\preceq$, is true iff $I$ is true. 
\end{compactenum}


The \emph{trivial dependency scheme} is one specific example of a permutation dependency scheme. This gives rise to the \emph{trivial dependency poset}, which sets $x\leq y$ whenever $x,y$ are in different quantifier blocks and $x$ is before $y$ in the prefix.
However, more refined permutation dependency schemes which give rise to other dependency posets are known to exist and can be computed efficiently~\cite{SamerSzeider09a,SlivovskySzeider16}. In particular, it is easy to verify that a dependency poset can be computed from any permutation dependency scheme in polynomial time (by computing the transitive closure).

To illustrate these definitions, consider the following QBF: 
\lv{
 $$
 \exists a \forall b \exists c (a\vee c) \wedge (b\vee c)
 $$}
\sv{
 $
 \exists a \forall b \exists c (a\vee c) \wedge (b\vee c)
 $.} 
 Then the trivial dependency poset would set $a\leq b \leq c$. However, for instance the resolution path dependency poset (arising from the resolution path dependency scheme~\cite{VanGelder11,SlivovskySzeider12}) contains a single relation $b\leq c$ (in this case, $a$ is incomparable to both $b$ and $c$).
 
\subsection{Q-resolution}
Q-resolution is a sound and complete resolution system for QBF~\cite{KleinebuningBubeck09}. Our goal here is to formalize the required steps for the Davis Putnam variant of Q-resolution.

We begin with a bit of required notation. For a QBF $I=(\phi, \tau)$ and a variable $x\in \var(\phi)$, let $\phi_x$ be the set of all clauses in $\phi$ containing the literal $x$ and similarly let $\phi_{\bar{x}}$ be the set of all clauses containing literal $\bar{x}$. We denote by $res(I,x)$ the QBF $I'=(\phi', \tau')$ such that $\tau'= \tau\setminus\{x\}$ and $\phi' = \phi \setminus (\phi_x \cup \phi_{\bar{x}}) \cup \{(D\setminus\{x\})\cup (C\setminus\{\bar{x}\})| D\in \phi_x; C \in \phi_{\bar{x}}\}$; informally, the two clause-sets are pairwise merged to create new clauses which do not contain $x$.
For a QBF $I=(\phi, \tau)$ and a variable $x\in \var(\phi)$ we denote by $I\setminus x$ the QBF $I=(\phi', \tau\setminus\{x\})$, where we get $\phi'$ from $\phi$ by removing all occurrences of $x$ and $\bar{x}$. 

\lv{\begin{LEM}}
\sv{\begin{LEM}[$\star$]}
\label{lem:DP-exists}
 Let $I=(\phi, \tau)$ and $x\in \var(\phi)$ be the last variable in $\tau$. If $x$ is existentially quantified, then $I$ is true if and only if $res(I,x)$ is true. 
\end{LEM}

\lv{
\begin{proof}
 Assume $I$ is true and let $\FFF$ be a winning strategy for existential player in $I$~\cite{KleinebuningBubeck09}. We will show that $\FFF$ is also a winning existential strategy in $res(I,x)$. Assume that the existential player played according to  $\FFF$ in $res(I, x)$, but there is a clause $B$ that is not satisfied at the end of the game. Clearly $B\in \{(D\setminus\{x\})\cup (C\setminus\{\bar{x}\})| D\in \phi_x; C \in \phi_{\bar{x}}\}$, otherwise $B$ is also a clause of $I$ and hence it has to be satisfied due to the existential player using $\FFF$. In particular, $B = (D\setminus\{x\})\cup (C\setminus\{\bar{x}\})$ for some $D\in \phi_x$ and $C \in \phi_{\bar{x}}$. Now if $\FFF$ assigns $x$ to $1$, since $\FFF$ is a winning strategy it follows that $C$ must be satisfied by some other literal, and hence $B$ must also be satisfied---a contradiction. A symmetric argument also leads to a contradiction if $\FFF$ assigns $x$ to $0$.
 
 Assume now that $res(I,x)$ is true and let $\FFF$ be a winning strategy for the existential player in $res(I,x)$. Now suppose that after all variables of $res(I,x)$ have been assigned according to the strategy $\FFF$, there is some $D\in \phi_x$ such that $D\setminus \{x\}$ is false. Since  $(D\setminus \{x\})\cup (C\setminus \{\bar{x}\})$ is true for all $C\in \phi_{\bar{x}}$, it means that all $C\in \phi_{\bar{x}}$ are true before we assign $x$, and our strategy can assign $x$ to $1$. On the other hand if $D\setminus \{x\}$ is true for all $D$ in $\phi_x$, our strategy assigns $x$ to $0$ and again satisfies all clauses of $I$. 
\end{proof}
}

\lv{\begin{LEM}}
\sv{\begin{LEM}[$\star$]}
\label{lem:DP-forall}
 Let $I=(\phi, \tau)$ and $x\in \var(\phi)$ be the last variable in $\tau$. If $x$ is universally quantified, then $I$ is true if and only if $I\setminus x$ is true. 
\end{LEM}
\lv{
\begin{proof}
 We will prove an equivalent statement: $I\setminus x$ is false if and only if $I$ is false. It is easy to see that if $\FFF$ is a winning strategy for the universal player in $I$, then if he plays according $\FFF$, then when the universal should assign the last variable $x$ there is either a clause that is already false and does not contain $x$, or a clause that contains $x$ and is false after an assignment of $x$ according to $\FFF$. In both cases $I\setminus x$ contains a clause that is false.
 
 On the other hand, assume $\FFF$ is a winning strategy for the universal player in $I\setminus x$ and the universal plays according to $\FFF$ in $I$ until all variables but $x$ are assigned. Clearly, this strategy leads to an assignment of variables such that there is a clause $B$ in $I\setminus x$, which is false under this assignment. From the definition of $I\setminus x$, it is easy to observe that $I$ contains a clause $B'$, which is equal to one of the following: $B$, $B\cup \{x\}$, or $B\cup \{\bar{x}\}$. It is straightforward to extend $\FFF$ in a way that whenever $\FFF$ falsifies a clause $B$, then the new strategy falsifies the clause $B'$.  
\end{proof}
}

\subsection{Treewidth}

Here we will introduce three standard characterizations of treewidth~\cite{Kloks94}:
tree decompositions, elimination orderings, and cops and
robber games. These will play an important role later on, when we define their counterparts in the dependency treewidth setting and use these in our algorithms.

\textbf{Tree decomposition:}  A tree decomposition of a graph $G$ is a pair $(T, \chi)$, 
  where $T$ is a rooted tree and $\chi$ is a function from
  $V(T)$ to subsets of $V(G)$, called a \emph{bag},
  such that the following properties hold:
  \begin{inparaitem}
  \item[(T1)] $\bigcup_{t\in V(T)}\chi(t)=V(G)$,
  \item[(T2)] for each $uv\in E(G)$ there exists $t\in V(T)$ such that
    $u,v\in \chi(t)$, and
  \item[(T3)] for every $u\in V(G)$, the set $T_{u}=\{t\in V(T): u\in
    \chi(t)\}$ induces a connected subtree of $T$.
  \end{inparaitem}

To distinguish between the vertices of the tree $T$ and the
vertices of the graph $G$, we will refer to the vertices of $T$ as
\emph{nodes}.
The \emph{width} of the tree decomposition $\mathcal{T}$ is
$\max_{t\in T}|\chi(t)|-1$. The \emph{treewidth} of $G$, $\tw(G)$, is the
minimum width over all tree decompositions of $G$. 
\lv{For a node $t \in V(T)$,
we denote by $T_t$ the subtree of $T$ rooted at $t$. 
The following fact will be useful later on:
\begin{PRO}\label{pro:tw-cut}
  Let $\mathcal{T}=(T,\chi)$ be a tree decomposition of a graph $G$
  and $t \in V(T)$ a node with parent $p$ in $T$. Then $\chi(p) \cap
  \chi(t)$ separates $\chi(T_t) \setminus \chi(p)$ from the rest of $G$.
\end{PRO}
}

\textbf{Elimination ordering:} An \emph{elimination ordering} of a graph is a linear order of its vertices.
Given an elimination ordering
$\phi$ of the graph $G$, the \emph{fill-in graph} $H$ of $G$
w.r.t. $\phi$ is the unique minimal graph such that: 
\vspace{-0.2cm}
\begin{itemize}
 \item $V(G)=V(H)$.
 \item $E(H)\supseteq E(G)$.
 \item If $0 \le k<i<j\le n$ and $v_i,v_j\in N_H(v_k)$, then $v_iv_j\in E(H)$.
\end{itemize}
The \emph{width} of elimination ordering $\phi$ is the maximum number
of neighbors of any vertex $v$ that are larger than $v$ (w.r.t.\ $\phi$) in $H$.


\textbf{(Monotone) cops and robber game:}
The \emph{cops and robber game} is played between two players (the cop-player
and the robber-player) on a graph $G$. A \emph{position} in the game is a pair
$(C,R)$ where $C \subseteq V(G)$ is the position of the cop-player and $R$ is a (possibly empty) connected component of $G
\setminus C$ representing the position of the robber-player. A move
from position $(C,R)$ to position $(C',R')$ is \emph{legal} if it satisfies
the following conditions:
\vspace{-0.2cm}
\begin{itemize}
\item[\noindent CM1] $R$ and $R'$ are contained in the same component of $G
  \setminus (C\cap C')$,
\item[\noindent CM2] $\guards(R) \subseteq C'$.
\end{itemize}
A play $\PPP$ is a sequence $(\emptyset, R_0), \dotsc, (C_n,R_n)$ of
positions such that for every $i$ with $1 \leq i < n$ it holds that the move
from $(C_i,R_i)$ to $(C_{i+1},R_{i+1})$ is legal; the cop-number of a play is $\max_{i\leq n}|C_i|$. 
A play $\PPP$ is won by the cop-player if $R_n=\emptyset$, otherwise
it is won by the robber-player. The cop-number of a strategy for the cop player is maximum cop-number over all plays that can arise from this strategy.
Finally, the cop-number of $G$ is the minimum cop-number of a winning strategy for the cop player.

\sv{
For any graph $G$ it holds that $G$ has treewidth $k$ iff $G$ has an elimination ordering of width $k$ iff $G$ has cop-number $k$~\cite{Kloks94}.
}

\lv{
\begin{proposition}[\cite{Kloks94}]
Let $G$ be a graph. The following three claims are equivalent:
\begin{itemize}
\item $G$ has treewidth $k$,
\item $G$ has an elimination ordering of width $k$,
\item $G$ has cop-number $k$.
\end{itemize}
\end{proposition}}

\section{Dependency Treewidth for QBF}
\newcommand{\sto}[1]{$#1$-elimination ordering}

We are now ready to introduce our parameter. We remark that in the case of dependency treewidth, it is advantageous to start with a counterpart to the elimination ordering characterization of classical treewidth, as this is used extensively in our algorithm for solving QBF. We provide other equivalent characterizations of dependency treewidth (representing the counterparts to tree decompositions and cops and robber games) in Section~\ref{sec:char}; these are not only theoretically interesting, but serve an important role in our algorithms for computing the dependency treewidth. 

Let $I = (\phi, \tau)$ be a QBF instance with a dependency poset $\PPP$. 
An elimination ordering of $G_I$ is \emph{compatible} with $\PPP$ if it is a linear extension of the reverse of $\PPP$; intuitively, this corresponds to being forced to eliminate variables that have the most dependencies first. For instance, if $\PPP$ is a trivial dependency poset then a compatible elimination ordering must begin eliminating from the rightmost block of the prefix. We call an elimination ordering of $G_I$ that is compatible with $\PPP$ a \emph{\sto{\PPP}} (or \emph{dependency elimination ordering}). The \emph{\stw{}} w.r.t.\ $\PPP$ is then the minimum width of a \sto{\PPP}.

\subsection{Using \stw}
Our first task is to show how dependency elimination orderings can be used to solve QBF.

\sv{\begin{THE}[$\star$]
\label{thm:using}
There is an algorithm that given \emph{1.}\ a QBF $I$ with $n$ variables and $m$ clauses, 
\emph{2.}\ a dependency poset $\PPP$ for I, and \emph{3.}\ a $\PPP$-elimination ordering $\pi$ of width $k$,
decides whether $I$ is true in time $\bigoh(3^{2k}kn)$.  Moreover, if $I$ is false, then the algorithm outputs a $Q$-resolution refutation of size $\bigoh(3^kn)$.
\end{THE}
}
\lv{\begin{THE}
\label{thm:using}
There is an algorithm that given
\begin{enumerate}
\item a QBF $I$ with $n$ variables and $m$ clauses, 
\item a dependency poset $\PPP$ for I, and 
\item a $\PPP$-elimination ordering $\pi$ of width $k$, 
\end{enumerate}
decides whether $I$ is true in time $\bigoh(3^{2k}kn)$.  Moreover, if $I$ is false, then the algorithm outputs a $Q$-resolution refutation of size $\bigoh(3^kn)$.
\end{THE}
}

\lv{\begin{proof}}
\sv{\begin{proof}[Sketch of Proof]}
Let $I=(\phi, \tau)$ and let $x_1,\dots, x_n$ denote the variables of $\phi$ such that $x_i\le_\pi x_{i+1}$ for all $1\le i<n$.  
From the definition of the dependency poset and the fact that $\pi$ is a dependency elimination ordering, it follows that the QBF instance $I'=(\phi, \tau')$, where $\tau'$ is the reverse of $\pi$, is true if and only if $I$ is true.
 
To solve $I$ we use a modification of the Davis Putnam resolution algorithm~\cite{DavisPutnam60}. We start with instance $I'$ and recursively eliminate the last variable in the prefix using Lemmas~\ref{lem:DP-exists} and~\ref{lem:DP-forall} until we either run out of variables or we introduce as a resolvent a non-tautological clause that is either empty or contains only universally quantified variables. We show that each variable we eliminate has the property that it only shares clauses with at most $k$ other variables, and in this case we introduce at most $3^k$ clauses of size at most $k$ at each step. 

From now on let $H$ be the fill-in graph of the primal graph of $I$ with respect to $\pi$, and let us define $I_i = (\phi^i, \tau^i)$ for $1\le i\le n$ as follows:
\sv{
(1) $I_1=I'$,
(2) $I_{i+1} = I_i\setminus x_i$ if $x_i$ is universally quantified, and
(3) $I_{i+1} = res(I_i, x_i)$, if $x_i$ is existentially quantified.
}

\lv{
\vspace{-0.2cm}
\begin{enumerate}
\item $I_1=I'$,
\item $I_{i+1} = I_i\setminus x_i$ if $x_i$ is universally quantified, and
\item $I_{i+1} = res(I_i, x_i)$, if $x_i$ is existentially quantified.
\end{enumerate}
}
Note that $x_i$ is always the last variable of the prefix of $I_i$ and it follows from Lemmas~\ref{lem:DP-exists} and~\ref{lem:DP-forall} that $I_{i+1}$ is true if and only if $I_i$ is true. Moreover, $I_n$ only contains a single variable, and hence can be decided in constant time. 
\lv{In the following, we show by induction 
}\sv{One can show by induction }that $I_{i+1}$ contains at most $3^k$ new clauses, i.e., clauses not contained in $I_i$.
To this end, we show and use the fact that both $\phi^i_x$ and $\phi^i_{\bar{x}}$ contain at most $3^k$ clauses, and this is sufficient to ensure a small Q-resolution refutation if the instance is false.
\sv{A formal proof of this fact and runtime analysis are provided in the full version.}
\lv{

\begin{CLM}
\label{cl:using}
 The instance $I_{i+1}$ contains at most $3^k$ clauses not contained in $I_i$.
 Furthermore, if the primal graph of $I_i$ is a subgraph of $H$, then so is the primal graph of $I_{i+1}$.
\end{CLM}
\begin{proof}[Proof of the Claim]

We distinguish two cases: either $x_{i}$ is universal, or existential. In the former case, it is easily observed that the primal graph of $I_i\setminus x_i$ is a subgraph of the primal graph of $I_i$ obtained by removing $x_i$ from the graph. Hence, in this case the primal graph of $I_{i+1}$ is a subgraph of $H$ as well. Moreover, we do not add any new clauses, only remove $x_i$ from already existing ones. 

In the later case, let $x_i$ be an existentially quantified variable. Let $x_px_q$ be an edge in the primal graph of $res(I_i,x_i)$, but not in the primal graph of $I_i$. Clearly, $x_p$ and $x_q$ are in a newly added clause $B=(D\setminus\{x\})\cup (C\setminus\{\bar{x}\})$ such that $ D\in \phi^i_{x_i}$ and $ C \in \phi^i_{\bar{x_i}}$. But that means
that both $x_p$ and $x_q$ are in a clause with $x_i$ in $I_i$. Hence, $x_px_i$ and $x_qx_i$ are both edges in the primal graph of $I_i$ and also in $H$. But, $p>i$ and $q>i$, since otherwise $x_p$ 
and $x_q$ will not appear in $I_{i+1}$. Since $H$ is the fill-in graph w.r.t.\ $\pi$, $H$ contains $x_px_q$ as well and $H$ is a supergraph of the primal graph of $I_i$. Moreover, as 
$\pi$ is an elimination ordering of the primal graph of $I$ of width $k$, there are at most $k$ variables in $\phi^i$ that appear with $x_i$ in a clause in $I_i$. Hence there are at 
most $3^k$ different clauses containing these variables, and so $\phi^i_{x_i}$, $\phi^i_{\bar{x_i}}$ contain at most $3^k$ clauses of size at most $k+1$. 
Finally, the set of new clauses $\{(D\setminus\{x_i\})\cup (C\setminus\{\bar{x_i}\})| D\in \phi^i_{x_i}; C \in \phi^i_{\bar{x_i}}\}$ contains at most $3^k$ clauses too, and the claim follows. 
\end{proof}
Since $I$ is equivalent to $I_n$, it is easy to see that if $I$ is false, then $I_n$ either contains the empty clause, or two clauses $\{x_n\}$ and $\{\bar{x_n}\}$, or $x_n$ is universally quantified. In these cases it is easy to see that the union of clauses of all $I_i$ (including the empty clause, if it is not already in $I_n$) is a Q-resolution refutation. Moreover, it follows from Claim~\ref{cl:using} that this Q-resolution refutation has at most $m+3^kn$ clauses.

For the runtime analysis, the time necessary to compute a variable-clause adjacency list data structure (which will be useful for the next step) is upper-bounded by $\bigoh(mk)$ due to clauses having size at most $k$ (since the width of the elimination ordering is at most $k$). If $x$ is universal, then $I\setminus x$ can be computed in time $\bigoh(3^k)$, since $x$ occurs in at most $3^k$ clauses. If $x$ is existential, then we need to compute the new clauses which takes time at most $\bigoh(3^{2k}\cdot k)$ since there are at most $3^{2k}$ pairs of clauses containing $x$ and each such clause has size at most $k$. 

Finally, observe that the total number of clauses cannot exceed $n\cdot k$ (because the treewidth of the matrix is also bounded by $k$), and so both $\bigoh(mk)$ and $m$ are superseded by the other term in the runtime and the resolution refutation size.
}
\end{proof}

\subsection{A Comparison of Decompositional Parameters for QBF}
\label{sub:comp}
As was mentioned in the introduction, two dedicated decompositional parameters have previously been introduced specifically for evaluating quantified Boolean formulas: \emph{prefix pathwidth} (and, more generally, \emph{prefix treewidth})~\cite{EibenGanianOrdyniak16} and \emph{respectful treewidth}~\cite{AtseriasOliva14}. The first task of this section is to outline the advantages of \stw~compared to these two parameters.

Prefix pathwidth is based on bounding the number of viable strategies in the classical two-player game characterization of the QBF problem~\cite{EibenGanianOrdyniak16}. As such, it decomposes the dependency structure of a QBF instance \emph{beginning from variables that have the least dependencies} (i.e., may appear earlier in the prefix). On the other hand, our \stw~is based on Q-resolution and thus decomposes the dependency structure \emph{beginning from variables that have the most dependencies} (i.e., may appear last in the prefix). Lemma~\ref{lem:incomp} shows that both approaches are, in principle, incomparable. That being said, \stw\ has two critical advantages over prefix treewidth/pathwidth:
\begin{enumerate}
\item \stw~outputs small resolution proofs, while it is not at all clear whether the latter can be used to obtain such resolution proofs;
\item \stw~supports a single-exponential fixed-parameter algorithm for QBF (Theorem~\ref{thm:using}), while the latter uses a prohibitive triple-exponential algorithm~\cite{EibenGanianOrdyniak16}.
\end{enumerate}

\lv{\begin{LEM}}
\sv{\begin{LEM}[$\star$]}
\label{lem:incomp}
Let us fix the trivial dependency poset. There exist infinite classes $\cal A,B$ of QBF instances such that:
\vspace{-0.2cm}
\begin{enumerate}
\item[a.] $\cal A$ has unbounded dependency treewidth but prefix pathwidth at most $1$;
\item[b.] $\cal B$ has unbounded prefix pathwidth (and prefix treewidth) but dependency treewidth at most~$1$.
\end{enumerate}
\end{LEM}

\sv{
\begin{proof}[Sketch of Proof]
Consider the following examples for $\cal A,B$. Let ${\cal A}=\{ A_i = \exists x_1, \dots, x_i \forall y \exists x$ $(y\vee x) \wedge \bigwedge_{j=1}^i (x_j\vee x)\}$, and let ${\cal B}=\{B_i=\exists x_1 \forall x_2 \exists x_3 \forall x_4 \dots \exists x_{2^i-1}\forall x_{2^i} \exists x_{2^i+1} \bigwedge_{j=1}^{i-1} ((x_j \vee x_{2j}) \wedge (x_j \vee x_{2j+1})\}$. It is straightforward to verify that these classes satisfy the conditions stipulated in the lemma.
\end{proof}
}

\lv{
\begin{proof}
\textbf{a.}
Let $$A_i = \exists x_1, \dots, x_i \forall y \exists x (y\vee x) \wedge \bigwedge_{j=1}^i (x_j\vee x).$$
The trivial dependency poset $\PPP_i$ for $A_i$ would be $\{x_1,\dots, x_i\} \le y \le x$. Hence every $\PPP_i$-elimination ordering must start with $x$, and then the width of such an ordering would be $i+1$. On the other hand, one can observe~\cite{EibenGanianOrdyniak16} that the path decomposition $\QQQ = (Q_1,\dots, Q_{i+1})$, where $Q_j= \{x_j,x\}$ for  $1\le j\le i$ and $Q_{i+1}=\{y,x\}$, is a prefix path-decomposition w.r.t. $\PPP_i$ of width $1$.

\textbf{b.} 
Consider the following formula with alternating prefix:
 $$B_i = \exists x_1 \forall x_2 \dots \forall x_{2^i} \exists x_{2^i+1} \bigwedge_{j=1}^{2^{i-1}} ((x_j \vee x_{2j}) \wedge (x_j \vee x_{2j+1}).$$
Since the quantifiers in the prefix of $B_i$ alternate, the trivial dependency poset $\PPP_i$ for $B_i$ would be 
the linear order $x_1 \le x_2\le \cdots \le x_{2^i}$. It is readily observed that the primal graph of $B_i$ is a balanced binary tree of depth $i$, and it is known that the pathwidth of such trees is $i-1$~\cite{Diestel95}. From the fact that pathwidth is a trivial lower bound for prefix pathwidth together with previous work on prefix treewidth~\cite[Theorem 6]{EibenGanianOrdyniak16}, it follows that $i-1$ is a lower bound on the prefix treewidth of $B_i$. 

On the other hand, since $\PPP_i$ is linear order, the only elimination ordering compatible with $\PPP_i$ is the reverse of $\PPP_i$. Moreover, from the definition of $B_i$, it is easily seen that $x_j$ has at most $1$ neighbor that is smaller w.r.t.\ $\PPP_i$, namely $x_{\lfloor j/2 \rfloor}$. Therefore, the \stw{} of $B_i$ is $1$. 
\end{proof}
}

Respectful treewidth coincides with \stw~when the trivial dependency scheme is used, i.e., represents a special case of our measure. Unsurprisingly, the use of more advanced dependency schemes (such as the resolution path dependency scheme~\cite{VanGelder11,SlivovskySzeider16}) allows the successful deployment of \stw\ on much more general classes of QBF instances. Furthermore, \stw\ with such dependency schemes will always be upper-bounded by respectful treewidth, and so algorithms based on \stw\ will outperform the previously introduced respectful treewidth based algorithms.

\sv{\begin{LEM}[$\star$]}
\lv{\begin{LEM}}
There exists an infinite class $\cal C$ of QBF instances such that $\cal C$ has unbounded dependency treewidth with respect to the trivial dependency poset but dependency treewidth at most $1$ with respect to the resolution-path dependency poset.
\end{LEM}

\sv{
\begin{proof}[Sketch of Proof]
It suffices to set $\cal C$ to be equal to the class $\cal A$ used in the proof of the previous lemma and then verify that $\cal C$ has the desired properties.
\end{proof}
}

\lv{
\begin{proof}
\lv{
Recall the previous example: $$A_i = \exists x_1, \dots, x_i \forall y \exists x (y\vee x) \wedge \bigwedge_{j=1}^i (x_j\vee x).$$}
\sv{
Recall the previous example: $A_i = \exists x_1, \dots, x_i \forall y \exists x (y\vee x) \wedge \bigwedge_{j=1}^i (x_j\vee x).$}
We have already established that $A_i$ has dependency treewidth $i+1$ when the trivial dependency poset is used. However, e.g., the resolution-path dependency poset~\cite{VanGelder11,SlivovskySzeider12}) contains a single relation $y\leq x$. Since the primal graph of $A_i$ is a star, it is then easy to verify that $A_i$ has dependency treewidth $1$ when the resolution-path dependency poset is used.
\end{proof}
}

Finally, we note that the idea of exploiting dependencies among variables has also given rise to similarly flavored structural measures in the areas of first-order model checking (first order treewidth)~\cite{AdlerWeyer12} and quantified constraint satisfaction (\hubiewidth\footnote{We remark that in their paper, the authors refer to their parameter simply as ``the width''. For disambiguation, here we call it CD-width (shorthand for Chen-Dalmau's width).})~\cite{ChenDalmau16}.
Even though the settings differ, Theorem~5.5~\cite{AdlerWeyer12} and Theorem~5.1~\cite{ChenDalmau16} can both be translated to a basic variant of Theorem~\ref{thm:using}. We note that this readily-obtained 
variant of Theorem~\ref{thm:using} would not account for dependency schemes. 
We conclude this subsection with two lemmas which show that there are classes of QBF instances that can be handled by our approach but are not covered by the results of 
Adler, Weyer~\cite{AdlerWeyer12} and Chen, Dalmau~\cite{ChenDalmau16}. 


\lv{
	\par Before we compare the dependency treewidth with the \hubiewidth{} of Chen and Dalmau, we will first define their parameter (or, more specifically, provide a translation into the QBF setting). \hubiewidth{} is also based on an elimination ordering, however with a slight modification: they also use the fact that we can eliminate universal variables without introducing new clauses (see Lemma~\ref{lem:DP-forall}) and therefore we can eliminate the last universal variable even if it appears together with many variables.  Formally the elimination ordering by Chen and Dalmau is defined as follows: 
	
	\begin{DEF}
	Let $I=(\phi, \tau)$  be a QBF instance. 
	Given a linear ordering $\preceq$ of $\var(\phi)$, 	
the \emph{CD-fill-in graph} $H_\preceq$ of $I$
	is the unique minimal graph such that: 
	\vspace{-0.2cm}
	\begin{itemize}
		\item $V(G_I)=V(H_\preceq)$.
		\item $E(H_\preceq)\supseteq E(G_I)$.
		\item If $u\preceq v \preceq w$, $u$ is existentially quantified, and $v, w \in N_{H_\preceq}(u)$, then $vw\in E(H_\preceq)$.
	\end{itemize}
	A \emph{CD-elimination ordering} of a QBF instance $I=(\phi, \tau)$ is 	
	  a linear order $\preceq$ of $\var(\phi)$ such that for each existentially quantified $x$ and each
	  universally quantified $y$ in $\var(\phi)$:
	\vspace{-0.2cm}
	  	  \begin{itemize}
	  \item If $y$ is before $x$ in $\tau$, then $x\preceq y$.
	  \item If $x$ is before $y$ in $\tau$ and there is an edge $xy$ in $H_\preceq$, then $y\preceq x$.
	  \end{itemize}
	The width of a CD-elimination ordering $\preceq$ is the maximum number of neighbors of an existentially quantified vertex $v$ that are larger than $v$ (w.r.t. $\preceq$) in $H_\preceq$.
	\end{DEF}
	
	We remark that one could also use the less restrictive form of elimination orderings considered above with any dependency poset, obtaining a more general variant of Theorem~\ref{thm:using}. However,
	such a notion would lose many of the nice structural properties used in our algorithms
	for finding the decompositions; for instance, the result is no longer a restriction of treewidth and
	does not have any immediate cops-and-robber game characterization. Hence finding such an ordering of small width would become a more challenging problem.
%
	}

\sv{\begin{LEM}[$\star$]}
	\lv{\begin{LEM}}
		There exist infinite classes $\cal D$, $\cal E$ of QBF instances such that:\begin{itemize}
			\item[a.] $\cal D$ has unbounded \hubiewidth{} but dependency treewidth at most $1$ w.r.t. the resolution-path dependency poset.
			\item[b.] $\cal E$ has unbounded dependency treewidth w.r.t.\ any dependency poset but \hubiewidth{} at most~$1$.
		\end{itemize} 
	\end{LEM}
	
	\sv{
		\begin{proof}[Sketch of Proof]
			It suffices to set $\cal D$ to be equal to the class $\cal A$ used in the proof of Lemma~\ref{lem:incomp} and then observe that the same class of instances is used as an example of a class with unbounded \hubiewidth{} by Chen, Dalmau~\cite[Example 3.6]{ChenDalmau16}. On the other hand, it easy to verify that the QBF instance  $E_i = \forall x_1\forall x_2 \cdots \forall x_i \bigwedge_{1\le p < q \le i}(x_p \vee x_q)$ has  \hubiewidth{} $0$, as $E_i$ does not contain an existential variable. However, the primal graph of $E_i$ is a clique and hence it has dependency treewidth $i-1$. 
		\end{proof}
	}
	
	\lv{
		\begin{proof}
			\textbf{a.}	Recall the previous example: $$A_i = \exists x_1, \dots, x_i \forall y \exists x (y\vee x) \wedge \bigwedge_{j=1}^i (x_j\vee x).$$
			We have already established that $A_i$ has dependency treewidth $1$ when  the resolution-path dependency poset is used. To establish the other directions it suffices to observe that Example 3.6 of Chen and Dalmau~\cite{ChenDalmau16} uses the class of prefixed graphs obtained from $\cal A$ by taking primal graph of instances $A_i$ and keeping the same prefix as an example of class with unbounded elimination width. 
			
			\textbf{b.} Let $$E_i = \forall x_1\forall x_2 \cdots \forall x_i \bigwedge_{1\le p < q \le i}(x_p \vee x_q).$$
			As $E_i$ does not contain any existentially quantified variable, it is easy to observe that \hubiewidth{} of $E_i$ is $0$. However, the primal graph of $E_i$ is a clique and hence dependency treewidth is $i-1$ w.r.t. every possible dependency poset. 
		\end{proof}
	}

\lv{
	Similarly as above, before we compare first order treewidth to dependency treewidth, we give some necessary definitions.
	
	Let $\cov$ be a binary relation on the variables of some QBF instance $I=(\phi, \tau)$. Then two
	variables $x$ and $y$ are entangled with respect to $\cov$ and $I$, if $x$ occurs in a clause with some variable $z$ such that  $y \cov z$ and $y$ occurs in a clause with some variable $z$ such that  $x \cov z$.
	
	\begin{DEF}
		Let $I=(\phi, \tau)$ be a QBF instance. Then  $\preceq_I^{\text{FO}}$ is the minimal (with respect to $\subseteq$) binary relation on $\var(\phi)$, such that the following holds:
		\begin{itemize}
			\item[(1)] $\preceq_I^{\text{FO}}$ is reflexive. 
			\item[(2)] $\preceq_I^{\text{FO}}$ is transitive. 
			\item[(3)] If $x$ is before $y$ in $\tau$, $x$ and $y$ have different quantifiers in $\tau$, and there is a sequence $x=z_0,\dotsc, z_n=y$ of variables such that for all $0\le i < n$ we have that $z_i$ and $z_{i+1}$ are entangled w.r.t. $\preceq_I$ and $I$ and that $x \preceq_I z_i$ or $y \preceq_I^{\text{FO}} z_i$, then $x\preceq_I^{\text{FO}} y$. 
		\end{itemize}
	\end{DEF}

	\begin{DEF}
		Let $I=(\phi,\tau)$ be a QBF instance  and $x\in \var(\phi)$. The essential alternation
		depth of $x$ in $I$, denoted by $\textup{ead}_I(x)$, is the maximum over all $\preceq_I^{\text{FO}}$-paths $P$ ending in $x$ of the
		number of quantifier changes in $P$, adding $+1$ in case the first variable on $P$ is existentially
		quantified and $+2$ if it is universally quantified.
	\end{DEF}
	For a QBF instance $I=(\phi,\tau)$, let $\VVV_I^{\text{FO}} = (\var(\phi), \leq_\text{ead}^I)$ be the poset such that $x\leq_\text{ead}^I y$ if and only if $\textup{ead}_I(x)\leq \textup{ead}_I(y)$. 
	Notice that $x \preceq_I^{\text{FO}} y$ implies $\textup{ead}_I(x)\leq \textup{ead}_I(y)$ and hence $\VVV_I^{\text{FO}}$ is an extension of $\preceq_I^{\text{FO}}$. 
	The first order treewidth is the minimal width of an elimination ordering that is compatible with $\VVV_I^{\text{FO}}$. 
	
	Now we are ready to prove the following lemma.
		}

\sv{\begin{LEM}[$\star$]}
	\lv{\begin{LEM}}
		There exists an infinite class $\cal F$ of QBF instances such that $\cal F$ has unbounded first order treewidth but dependency treewidth at most $2$ with respect to the resolution-path dependency poset.
	\end{LEM}
	
	\sv{
		\begin{proof}[Sketch of Proof]
			Let $
			F_i =\forall x_{2i} y_{2i} \exists x_{2i-1} y_{2i-1} \cdots \forall x_{2} y_{2} \exists x_{1} y_{1} \forall z (z\vee x_1\vee y_1) 
			\bigwedge_{j=1}^{2i-1} [(x_j\vee x_{j+1}) \wedge  (y_j\vee y_{j+1})] \wedge \bigwedge_{j=1}^{i} (x_{2j-1}\vee y_{1})
			$.	
			It is readily observed that the elimination ordering $zx_1x_2\dotsc x_{2i}y_1y_2\dotsc y_{2i}$ of width $2$ is compatible with the resolution-path dependency poset for this formula. On the other hand, the elimination ordering obtained from first order treewidth is forced to eliminate $y_1$ before $x_j$, $j\ge 2$~(see Definitions~3.3, 3.10, 3.15, together with the definitions on page 5 of Adler and Weyer~\cite{AdlerWeyer12}). Therefore, the first order treewidth of this instance would be at least $i-1$.
		\end{proof}
	}
	
	\lv{
		\begin{proof}
			Let 
			\begin{align*}
			F_i =\forall x_{2i} y_{2i} \exists  x_{2i-1}& y_{2i-1}  \cdots \forall x_{2} y_{2} \exists x_{1} y_{1} \forall z \\ &  (z\vee x_1\vee y_1)  
			\bigwedge_{j=1}^{2i-1} [(x_j\vee x_{j+1}) \wedge  (y_j\vee y_{j+1})] \wedge \bigwedge_{j=1}^{i}  (x_{2j-1}\vee y_{1}). 
			\end{align*}
			
			The resolution-path dependency poset would give us following relations: a chain $y_{2i}\le \cdots \le y_{2} \le y_1 \le z$ and a chain $x_{2i}\le \cdots \le x_{2} \le x_1 \le z$. It is readily observed that the elimination ordering $zx_1x_2\dotsc x_{2i}y_1y_2\dotsc y_{2i}$ with width $2$ is compatible with this dependency poset. 
			
			On the other hand, one can observe that $x_{2j}\preceq_{F_i}^{\text{FO}} x_{2j-1}$ and that $y_1$ and $x_{2j-1}$ are entangled w.r.t. $\preceq_{F_i}^{\text{FO}}$ and hence also $x_{2j}\preceq_{F_i}^{\text{FO}} y_{1}$ for all $1\le j\le i$. It follows that any elimination ordering that is compatible with $\VVV_I^{\text{FO}}$ must have $y_1$ before all $x_j$ for $2\le j\le 2i$ and the width of such elimination ordering is at least $i-1$. 
		\end{proof}
		
Finally, we remark that if one were to show that $\VVV_I^{\text{FO}}$ is a dependency poset, then this would imply that first order treewidth is a special case of dependency treewidth (for this dependency poset). However, proving such a claim goes beyond the scope of this paper.
	}

\section{Dependency Treewidth: Characterizations}
\label{sec:char}

In this section we obtain other equivalent characterizations of dependency treewidth. The purpose of this endeavor is twofold. From a theoretical standpoint, having several natural characterizations (corresponding to the characterizations of treewidth) is not only interesting but also, in some sense, highlights the solid foundations of a structural parameter. From a practical standpoint, the presented characterizations play an important role in Section~\ref{sec:finding}, which is devoted to algorithms for finding optimal dependency elimination orderings.

\textbf{Dependency tree decomposition:}
Let $I$ be a QBF instance with primal graph $G$ and dependency poset $\PPP$ and let
$(T,\chi)$ be a tree decomposition of $G$. Note that the rooted tree
$T$ naturally induces a partial order $\leq_T$ on its nodes, where the
smallest element is the root and leaves form maximal elements. For a vertex $v \in V(G)$, we denote by
$\ptdt{T}{v}$ the unique $\leq_T$\hy minimal node $t$ of $T$ with $v \in
\chi(t)$, which is well-defined because of Properties (T1) and (T3) of
a tree decomposition. Let $<_{\mathcal{T}}$ be the partial ordering of
$V(G)$ such that $u <_{\mathcal{T}} v$ if and only if $\ptdt{T}{u} <_T
\ptdt{T}{v}$ for every $u,v \in V(G)$. We say that $(T,\chi)$ is a \emph{dependency tree
  decomposition} if it satisfies the following additional property:
\begin{itemize}
\item[(T4)] $<_{\mathcal{T}}$ is compatible with $\leq^\PPP$, i.e.,
  for every two vertices $u$ and $v$ of $G$ 
  it holds that whenever $\ptdt{T}{u} <_T \ptdt{T}{v}$ then it does not
  hold that $v \leq^\PPP u$. 
\end{itemize}

\begin{LEM}
  A graph $G$ has a \sto{\PPP} of width at most $\omega$ if and only if
  $G$ has a dependency tree decomposition of width at most $\omega$.
  Moreover, a \sto{\PPP} of width $\omega$ can be obtained
  from a dependency tree decomposition of width $\omega$ in
  polynomial-time and vice versa.
\end{LEM}
\begin{proof}
  For the forward direction we will employ the construction given
  by Kloks in~\cite{Kloks94}, which shows that a normal elimination
  ordering can be transformed into a tree decomposition of the same
  width. We will then show that this construction also retains the
  compatibility with $\PPP$.
  Let $\leq^\phi=(v_1, \dotsc, v_n)$ be a dependency elimination ordering for $G$
  of width $\omega$ and let $H$ be the fill-in graph of $G$
  w.r.t. $\leq^\phi$. We will iteratively construct a sequence
  $(\mathcal{T}_0,\dotsc, \mathcal{T}_{n-1})$ such that for every $i$ with
  $0 \leq i < n$, $\mathcal{T}_i=(T_i,\chi_i)$ is dependency tree decompositions
  of the graph $H_i=H[\{v_{n-i}, \dotsc, v_n\}]$ of width at most
  $\omega$. Because $\mathcal{T}_{n-1}$ is a dependency tree
  decomposition of $H_{n-1}=H$ of width at most $\omega$, this 
  shows the forward direction of the lemma.
  In the beginning we set $\mathcal{T}_0$ to be the trivial
  tree decomposition of $H_0$, which contains merely one node 
  whose bag consists of the vertex $v_n$. Moreover, for every $i$ with
  $0 < i < n$, $\mathcal{T}_i$ is obtained from $\mathcal{T}_{i-1}$ as
  follows. Note that because $N_{H_{i}}(v_{n-i})$ induces a clique in
  $H_{i-1}$, $\mathcal{T}_{i-1}$ contains a node that covers all
  vertices in $N_{H_{i}}(v_{n-i})$. 
  Let $t$ be any such bag, then
  is $\mathcal{T}_i$ is obtained from $\mathcal{T}_{i-1}$ by adding a
  new node $t'$ to $T_{i-1}$ making it adjacent to $t$ and setting
  $\chi_i(t')=N_{H_{i}}[v_{n-i}]$. 
  It is known~\cite{Kloks94} that $\mathcal{T}_i$ satisfies the
  Properties (T1)--(T3) of a tree decomposition and it hence only
  remains to show that $\mathcal{T}_i$ satisfies (T4). Since, by induction hypothesis, $\mathcal{T}_{i-1}$ is a dependency tree decomposition, Property~(T4)
  already holds for every pair $u,v \in V(H_{i-1})$. Hence it only
  remains to consider pairs $u$ and $v_{n-i}$ for some $u \in
  V(H_{i-1})$. Because the only node containing $v_{n-i}$ in
  $\mathcal{T}_i$ is a leaf, we can assume that $\ptdt{T}{u}
  <_T \ptdt{T}{v_{n-i}}$ and because $v_{n-i} \leq^\phi u$ it cannot hold that
  $v_{n-i} \leq^\PPP u$, as required.

  For the reverse direction, let $\mathcal{T}=(T,\chi)$ be a $\PPP$\hy
  tree decomposition of $G$ of width at most $\omega$. 
  It is known~\cite{Kloks94} that any linear extension of
  $<_{\mathcal{T}}$ is an elimination ordering for $G$ of width at most $\omega$.
  Moreover, because of Property~(T4), $<_{\mathcal{T}}$ is compatible
  with $\leq^\PPP$ and hence there is a linear extension of
  $<_{\mathcal{T}}$, which is also a linear extension of the reverse of $\leq^\PPP$.
\end{proof}

\textbf{Dependency cops and robber game:} 
Recalling the definition of the (monotone) cops and robber game for treewidth, we define the \emph{dependency cops and robber game} (for a QBF instance $I$ with dependency poset $\PPP$) analogously but with the additional restriction that legal moves must also satisfy a third condition:
\vspace{-0.1cm}
\begin{itemize}
\item[\noindent CM3] $C' \setminus C$ is downward-closed in $R$,
i.e., there is no $r \in R\setminus C'$ with $r \leq^\PPP c$ for any $c \in C' \setminus C$. \end{itemize}
Intuitively, condition CM3 restricts the cop-player by forcing him to search vertices (variables) in an order that is compatible with the dependency poset. 

To formally prove the equivalence between the cop-number for this restricted game and dependency treewidth, we will need to also formalize the notion of a strategy. Here we will represent strategies for the cop-player as rooted trees whose nodes are labeled with positions for
the cop-player and whose edges are labeled with positions for the
robber-player. Namely, we will represent winning strategies for the
cop-player on a primal graph $G$ by a triple $(T,\alpha,\beta)$, where $T$ is a
rooted tree, $\alpha : V(T) \rightarrow 2^{V(G)}$ is a mapping from
the nodes of $T$ to subsets of
$V(G)$, and $\beta : E(T) \rightarrow 2^{V(G)}$, satisfying the
following conditions:
  \vspace{-0.1cm}
\begin{itemize}
\item[\noindent \fbox{CS1}] $\alpha(r)=\emptyset$ and for every component $R$ of $G$, the root node $r$ of $T$
 has a unique child $c$ with $\beta(\{r,c\})=R$, and
\item[\noindent \fbox{CS2}] for every other node $t$ of $T$ with parent $p$ it holds that: 
  the move from position
  $(\alpha(p),\beta(\{p,t\}))$ to position
  $(\alpha(t),\beta(\{t,c\}))$ is legal for every child $c$ of $t$
  and moreover for every component $R$ of $G \setminus
  \alpha(t)$ contained in $\beta(\{p,t\})$, $t$ has a unique child
  $c$ with $\beta(\{t,c\})=R$.
\end{itemize}
Informally, the above properties ensure that every play consistent
with the strategy is winning for the cop-player and moreover for every
counter-move of the robber-player, the strategy gives a move for the
cop-player. The width of a winning strategy for the cop-player is the
maximum number of cops simultaneously placed on $G$ by the
cop-player, i.e., $\max_{t \in V(T)} |\alpha(t)|$. The cop-number of a
graph $G$ is the minimum width of any winning strategy for the
cop-player on $G$. We are now ready to show the equivalence between
dependency tree decompositions and winning strategies for the cop-player.
\lv{\begin{LEM}}
\sv{\begin{LEM}[$\star$]}
  For every graph $G$ the width of an optimal dependency tree
  decomposition plus one is equal to the cop-number of the graph. Moreover, 
  a dependency tree decomposition of width $\omega$ can be obtained
  from a winning strategy for the cop-player of width $\omega+1$ in
  polynomial-time and vice versa.
\end{LEM}
\lv{\begin{proof}}
\sv{\begin{proof}[Sketch of Proof]}
  Let $\mathcal{T}=(T,\chi)$ be a dependency tree decomposition of $G$
  of width $\omega$. We start by showing that $\mathcal{T}$ can be
  transformed into a dependency tree decomposition of width $\omega$
  satisfying:
    \vspace{-0.1cm}
  \begin{itemize}
  \item[(*)] $\chi(r)=\emptyset$ for the root node $r$ of $T$ and
    for every node $t \in V(T)$ with child $c \in V(T)$ in $T$
    the set $\chi(T_c) \setminus \chi(t)$ is a component of $G
    \setminus \chi(t)$.
  \end{itemize}
  To ensure that $\mathcal{T}$ satisfies $\chi(r)=\emptyset$ it is
  sufficient to add a new root vertex $r'$ to $T$ and set
  $\chi(r')=\emptyset$. We show next that starting from the root of $T$
  we can ensure that for every node $t \in V(T)$ with
  child $c$ the set $\chi(T_c) \setminus \chi(t)$ is a component of $G
  \setminus \chi(t)$. Let $t$ be a node with child $c$ in $T$ for
  which this does not hold. 
  \lv{Because of Proposition~\ref{pro:tw-cut},}\sv{By the well-known separation property of tree-decompositions,}
  we have that $\chi(T_c) \setminus \chi(t)$ is a set of components,
  say containing $C_1,\dotsc, C_l$, of
  $G \setminus \chi(t)$. For every $i$ with $1\leq i \leq l$, let
  $\mathcal{T}_i=(T_i,\chi_i)$ be the dependency tree decomposition with $T_i=T_c$ and $\chi_i(t')=\chi(t')
  \cap (C_i \cup \chi(t))$ and root $r_i=c$. Then we replace the
  entire sub dependency tree decomposition of $\mathcal{T}$ induced by $T_c$ in
  $T$ with the tree decompositions
  $\mathcal{T}_1,\dotsc,\mathcal{T}_l$ such that $t$ now becomes
  adjacent to the roots $r_1,\dotsc,r_l$. It is straightforward to
  show that the result of this operation is again a dependency tree
  decomposition of $G$ of width at most $\omega$ and moreover the node
  $t$ has one child less that violates (*). By iteratively applying
  this operation to every node $t$ of $\mathcal{T}$ we eventually
  obtain a dependency tree decomposition that satisfies (*).

  Hence w.l.o.g. we can assume that $\mathcal{T}$ satisfies (*). We
  now claim that $(T,\alpha,\beta)$ where:
  \vspace{-0.2cm}
  \begin{itemize}
  \item $\alpha(t)=\chi(t)$ for every $t \in V(T)$,
  \item for a node $t \in V(T)$ with parent $p \in V(T)$, 
    $\beta(\{p,t\})=\chi(T_t) \setminus \chi(p)$.
  \end{itemize}
  is a winning strategy for $\omega+1$ cops. Observe that because
  $\mathcal{T}$ satisfies (*), it holds that $\alpha(r)=\emptyset$ and
  for every $t \in V(T)$
  with parent $p \in V(T)$, the pair $(\alpha(p),\beta(\{p,t\})$ is a
  position in the visible $\PPP$\hy cops and robber game on $G$. 
  \sv{Furthermore, it is possible to verify that for every $t$, $p$ as above and every child $c$ of
  $t$ in $T$, it holds that the move from $(\alpha(p),\beta(\{p,t\})$
  to $(\alpha(t),\beta(\{t,c\})$ is valid. 
  }
    
  \lv{ We
  show next that for every $t$, $p$ as above and every child $c$ of
  $t$ in $T$, it holds that the move from $(\alpha(p),\beta(\{p,t\})$
  to $(\alpha(t),\beta(\{t,c\})$ is valid. 
  Because 
  $\beta(\{t,c\})=\chi(T_c) \setminus \chi(t) \subseteq \chi(T_t)
  \setminus \chi(p)=\beta(\{p,t\})$ and $\beta(\{p,t\})$ is connected
  in $G \setminus \alpha(p)$ (and hence also in $G \setminus
  (\alpha(p)\cap\alpha(t))$), it follows that
  $\beta(\{p,t\})$ and $\beta(\{t,c\})$ are contained in the same
  component of $G \setminus (\alpha(p)\cap\alpha(t))$, which shows
  CM1. Because of Proposition~\ref{pro:tw-cut}, it holds that
  $\alpha(p) \cap \alpha(t)$ and hence in particular $\alpha(t)$
  separates $\beta(\{p,t\})$ from the rest of the graph, which shows 
  CM2, i.e, $\guards(\beta(\{p,t\}))\subseteq \alpha(t)$. 
  Towards
  showing CM3 suppose for a contradiction that there is a $r \in
  \beta(\{p,t\}) \setminus \alpha(t)$ with $r \leq^\PPP c$ for some $c
  \in \alpha(t) \setminus \alpha(p)$. Because $r \in \beta(\{p,t\})
  \setminus \alpha(t)$ and $c \in \alpha(t)$, we obtain that $c
  <_{\mathcal{T}} r$ and because of (T4) it follows that $r \leq^\PPP
  c$ does not hold. 
  Note that CS1 and the second part of CS2, i.e., for every node $t$
  with parent $p$ and every component $R$ of $G \setminus \alpha(t)$,
  $t$ has a unique child $c$ with $\beta(\{t,c\})=R$, both hold
  because $\mathcal{T}$ satisfies T1 and T3 (T3 is only needed to show
  that the child is unique).}

  On the other hand, let $\mathcal{S}=(T,\alpha,\beta)$ be a winning strategy for the
  cop-player in the visible $\PPP$-cops and robber game on $G$ using
  $\omega$ cops. 
  Observe that $\mathcal{S}$ can be
  transformed into a winning strategy for the cop-player using
  $\omega$ cops satisfying:
  \vspace{-0.2cm}
  \begin{itemize}
  \item[($^a$)] for every node $t$ of $T$ with parent $p$ it holds that
    $\alpha(t) \subseteq \guards(\beta(\{p,t\}))\cup \beta(\{p,t\})$.
  \end{itemize}
  Indeed; if ($^a$) is violated, then one can
  simply change $\alpha(t)$ to $\alpha(t) \cap
  (\guards(\beta(\{p,t\}))\cup \beta(\{p,t\}))$ without violating any
  of CS1 or CS2. Hence we can assume that $\mathcal{S}$ satisfies ($^a$).

 We now claim that $\mathcal{T}=(T,\alpha)$ is a dependency tree decomposition of $G$
  of width $\omega-1$.
  Towards showing T1, let $v \in V(G)$. Because of CS1, it holds that
  either $v \in \alpha(r)$ for the root $r$ of $T$ or there is a child
  $c$ of $r$ in $T$ with $v \in \beta(\{r,c\})$. Moreover, due to CS2
  we have that either $v \in \alpha(c)$ or $v \in \beta(\{c,c'\})$ for
  some child $c'$ of $c$ in $T$. By proceeding along $T$, we
  will eventually find a node $t \in V(T)$ with $v \in \alpha(t)$. 
  Towards showing T2, let $\{u,v\} \in E(G)$. Again because of CS1, it
  holds that either $\{u,v\} \subseteq \alpha(r)$, or $\{u,v\}
  \subseteq \guards(\beta(\{r,c\})) \cup \beta(\{r,c\})$ for some
  child $c$ of $r$ in $T$. Because of CM2, we obtain that
  $\guards(\beta(\{r,c\})) \subseteq \alpha(c)$ and together with CS2,
  we have that either $\{u,v\} \subseteq \alpha(c)$ or $\{u,v\}
  \subseteq \guards(\beta(\{c,c'\}))\cup \beta(\{c,c'\})$ for some
  child $c'$ of $c$ in $T$. By proceeding along $T$, we will
  eventually find a node $t \in V(T)$ with $\{u,v\} \in \alpha(t)$.
    \sv{
  Finally, in order to argue that T3 and T4 hold, we will first establish that $\mathcal{S}$ satisfies the following property:
  \vspace{-0.2cm}
  \begin{itemize}
  \item[$(^b)$] for every node $t$
    with child $c$ in $T$ it holds that $\bigcup_{t' \in
    V(T_c)}\alpha(t') \subseteq \guards(\beta(\{t,c\})\cup\beta(\{t,c\})$. 
  \end{itemize}
  Because of CM2 we have that $\beta(\{t,c\}) \subseteq
  \beta(\{p,t\})$ for every three nodes $p$, $t$, and $c$ such that
  $p$ is the parent of $t$ which in turn is the parent of $c$ in
  $T$. Moreover, because of $(^a)$ we have that $\alpha(t) \subseteq
  \guards(\beta(\{p,t\}))\cup \beta(\{p,t\})$ for every node $t$ with
  parent $p$ in $T$. Applying these two facts iteratively along a path
  from $t$ to any of its descendants $t'$ in $T$, we obtain that
  $\alpha(t') \subseteq \guards(\beta(\{p,t\}))\cup \beta(\{p,t\})$,
  as required. 

  Knowing $(^b)$ and $(^a)$, it is not too difficult to show that T3 and T4 hold. This means that $\mathcal{T}$ is a dependency tree-decomposition, completing the proof. 
  }
    \lv{ 
  Next, we prove that T3 holds by showing that for any three nodes $t_1$, $t_2$, and
  $t_3$ in $T$ such that $t_2$ lies on the unique path from $t_1$ to
  $t_3$ in $T$, it holds that $\alpha(t_1)\cap \alpha(t_2)\subseteq
  \alpha(t_3)$. We will distinguish two cases: (1) $t_1$ is not an ancestor
  of $t_3$ and vice versa, and (2) $t_1$ and $t_2$ lie on
  the unique path from the root of $T$ to $t_3$.

  Before proceeding, we need to establish that $\mathcal{S}$ satisfies the following property:
  \vspace{-0.2cm}
  \begin{itemize}
  \item[$(^b)$] for every node $t$
    with child $c$ in $T$ it holds that $\bigcup_{t' \in
    V(T_c)}\alpha(t') \subseteq \guards(\beta(\{t,c\})\cup\beta(\{t,c\})$. 
  \end{itemize}
  Because of CM2 we have that $\beta(\{t,c\}) \subseteq
  \beta(\{p,t\})$ for every three nodes $p$, $t$, and $c$ such that
  $p$ is the parent of $t$ which in turn is the parent of $c$ in
  $T$. Moreover, because of $(^a)$ we have that $\alpha(t) \subseteq
  \guards(\beta(\{p,t\}))\cup \beta(\{p,t\})$ for every node $t$ with
  parent $p$ in $T$. Applying these two facts iteratively along a path
  from $t$ to any of its descendants $t'$ in $T$, we obtain that
  $\alpha(t') \subseteq \guards(\beta(\{p,t\}))\cup \beta(\{p,t\})$,
  as required.

  Towards showing case (1) assume that this is not the case, i.e., there are
  $t_1$, $t_2$, and $t_3$ as above and a vertex $v \in \alpha(t_1)
  \cap \alpha(t_3)$ but $v \notin \alpha(t_2)$. W.l.o.g. we can assume
  that $t_2$ is the least common ancestor of $t_1$ and $t_2$ in $T$,
  otherwise we end up in case (2). Let $c_1$ and $c_2$ be the two children
  of $t_2$ such that $t_1 \in V(T_{c_1})$ and $t_2 \in
  V(T_{c_2})$. Then because $v \in \alpha(t_i)$ we obtain from $(^b)$
  that $v \in \guards(\beta(\{t_2,c_i\})) \cup \beta(\{t_2,c_i\}))$
  for any $i \in \{1,2\}$. Because the sets $\beta(\{t_2,c_1\})$ and
  $\beta(\{t_2,c_2\})$ are disjoint, it follows that $v \in
  \guards(\beta(\{t_2,c_i\}))$ and thus $v \in \alpha(t_2)$
  contradicting our assumption that this is not the case.

  Towards showing case (2) assume that this is not the case, i.e., there are
  $t_1$, $t_2$, and $t_3$ as above and a vertex $v \in \alpha(t_1)
  \cap \alpha(t_3)$ but $v \notin \alpha(t_2)$. W.l.o.g. we can assume
  that $t_2$ is a child of $t_1$ in $T$. Because $v \in \alpha(t_3)$
  we obtain from $(^b)$ that $\alpha(t_3) \subseteq \guards(\beta(\{t_1,t_2\}) \cup
  \beta(\{t_1,t_2\})$. Hence either $v \in
  \guards(\beta(\{t_1,t_2\}))$ or $v \in \beta(\{t_1,t_2\})$. In the
  former case it follows from CM2 that $v \in \alpha(t_2)$ a
  contradiction to our assumption that $v \notin \alpha(t_2)$ and in
  the later case, we obtain that $v \notin \alpha(t_1)$ a
  contradiction to our assumption that $v \in \alpha(t_1)$.

  Towards showing T4 assume that this is not the case, i.e., there are
  $u,v \in V(G)$ with $u <_{\mathcal{T}} v$ but $v \leq^\PPP u$. Let
  $p$ be the parent of $\ptdt{T}{u}$ in $T$. Then because $u,v \in
  \bigcup_{t' \in V(T_{\ptdt{T}{u})})}\alpha(t')$ and $(^b)$, we obtain that 
  $u,v \in \guards(\beta(\{p,\ptdt{T}{u}\})) \cup
  \beta(\{p,\ptdt{T}{u}\})$. Moreover since neither $u$ or $v$ are
  in $\alpha(p)$, we have that $u,v \in
  \beta(\{p,\ptdt{T}{u}\})$. Finally because $u \in
  \alpha(\ptdt{T}{u})$ but not $v \in \alpha(\ptdt{T}{u})$, we obtain
  from CM3 that $u \leq^\PPP v$ contradicting our assumption that $v
  \leq^\PPP u$.
  }
\end{proof}

\section{Computing Dependency Treewidth}
\label{sec:finding}

In this section we will present two exact algorithms to compute
dependency treewidth. The first algorithm is based on the
characterization of dependency treewidth in terms of the cops and
robber game and shows that, for every fixed $\omega$, determining
whether a graph has dependency treewidth at most $\omega$, and in the positive case
also computing a dependency tree decomposition of width at most $\omega$,
can be achieved in polynomial time. The second algorithm is based on a
chain partition of the given dependency poset and shows that if the
width of the poset is constant, then an optimal dependency tree
decomposition can be constructed in polynomial time. 

Before proceeding to the algorithms, we would like to mention here that the fixed-parameter
algorithm for computing first order treewidth~\cite{AdlerWeyer12} can also be used for computing dependency treewidth in the restricted case that the trivial dependency poset is used. 

Below we provide our first algorithm for the computation of dependency treewidth.

\lv{\begin{THE}}
\sv{\begin{THE}[$\star$]}
\label{thm:comput}
   There is an algorithm running in time $\bigoh(|V(G)|^{2\omega+2})$ that, given a graph $G$
   and a poset $\PPP=(V(G),\leq^\PPP)$ and $\omega\in \Nat$, determines whether $\omega$
   cops have a winning strategy in the dependency cops and
   robber game on $G$ and $\PPP$, and if so outputs such a winning strategy.
\end{THE}
\sv{\begin{proof}[Sketch of Proof]}
\lv{\begin{proof}}
  The idea is to transform the cops and robber game on $G$ into a much
  simpler two player game, the so-called simple two player game, 
  which is played on all possible positions of the cops and robber
  game on $G$. 

  A \emph{simple two player game} is  played between two players,
  which in association to the cops and robber game, we will just call
  the cops and the robber player~\cite{GraedelThomasWilke02}. Both players play by moving a token
  around on a so-called \emph{arena}, which is a triple
  $\mathcal{A}=(V_C,V_R,A)$ such that $((V_C \cup V_R),A)$ is a bipartite
  directed graph with bipartition $(V_C,V_R)$. The vertices in $V_C$ are
  said to belong to the cop-player and the vertices in $V_R$ are said to
  belong to the robber-player. Initially, one token is placed 
  on a distinguished starting vertex $s \in V_C \cup V_R$. From then
  onward the player who owns the vertex, say $v$, that currently
  contains the token, has to move the token to an arbitrary successor (i.e., out-neighbor)
  of $v$ in $\mathcal{A}$. The cop-player wins if the robber-player
  gets stuck, i.e., the token ends up in a vertex owned by the
  robber-player that has no successors in $\mathcal{A}$, otherwise the
  robber-player wins. It is well-known that strategies in this game
  are deterministic and memoryless, i.e., strategies for a player are simple
  functions that assign every node owned by the player one of its successors.
  Moreover, the winning region for both players
  as well as their corresponding winning strategy can be computed in
  time $\bigoh(|V_C\cup V_R|+|A|)$ by the following algorithm. The algorithm first computes the winning region $W_C$, as follows.
  
  Initially all vertices owned by the robber-player which do
  not have any successors in $\mathcal{A}$ are placed in $W_C$. 
  The algorithm then iteratively adds the following vertices to $W_C$:
  \vspace{-0.2cm}
  \begin{itemize}
  \item all vertices owned by the cop-player that have at least one
    successor $W_C$,
  \item all vertices owned by the robber-player for which all
    successors are in $W_C$.
  \end{itemize}
  Once the above process stops, the set $W_C$ is the winning region of
  the cop-player in $\mathcal{A}$ and $(V_C\cup V_R)\setminus W_C$ is the
  winning region for the robber-player. Moreover, the winning
  strategy for both players can now be obtained by choosing for every
  vertex a successors that is in the winning region of the player
  owning that vertex (if no such vertex exists, then an arbitrary successor
  must be chosen).

  Given a graph $G$, a poset $\PPP=(V(G),\leq^\PPP)$, and an integer
  $\omega$, we construct an arena $\mathcal{A}=(V_C,V_R,A)$ and a starting
  vertex $s \in V_R$ such that $\omega$ cops have a winning strategy in
  the $\PPP$\hy cops and robber game on $G$ iff the
  cop-player wins from $s$ in the simple two player game on
  $\mathcal{A}$ as follows:
    \vspace{-0.2cm}
  \begin{itemize}
  \item We set $V_C$ to be the set of all pairs $(C,R)$ such that
    $(C,R)$ is a position in the $\PPP$-cops and robber game
    on $G$ using at most $\omega$ cops ($|C|\leq \omega$),
  \item We set $V_R$ to be the set of all triples $(C,C',R)$ such
    that:
    \begin{itemize}
    \item $(C,R)$ is a position in the $\PPP$-cops and robber game
      on $G$ using at most $\omega$ cops\lv{ ($|C|\leq \omega$)}, and
    \item $C' \subseteq V(G)$ is a potential new cop-position for
      at most $\omega$ cops from $(C',R')$, i.e.,
      $\guards(R)\subseteq C'$ and $C$, $R$, and $C'$ satisfy CM3.
    \end{itemize}
  \item From every vertex $(C,R) \in V_C$ we add an arc to all
    vertices $(C,C',R) \in V_R$.
  \item From every vertex $(C,C',R) \in V_R$ we add an arc to all
    vertices $(C',R') \in V_C$ such that the move from $(C,R)$ to
    $(C',R')$ is legal.
  \item Additionally $V_R$ contains the starting vertex $s$ that has
    an outgoing arc to every vertex $(\emptyset,R) \in V_C$ such that
    $R$ is a component of $G$.
      \vspace{-0.2cm}
  \end{itemize}
  
\sv{Finally, it is straightforward to show that the cop-player has a winning strategy from $s$ in $\cal A$ iff $G$ and $\PPP$ have cop-number at most $\omega$. 
}  
  \lv{
  By construction $|V_C|\leq |V(G)|^{\omega+1}$ and $|V_R|\leq
  |V(G)|^{2\omega+1}$. Moreover, because every vertex in $V_C$ has at
  most $|V(G)|^{\omega}$ successors and every vertex in $V_R$ has at
  most $|V(G)|$ successors, we obtain that $|A|\leq |V(G)|^{2\omega+2}$.
  Let us now analyze the running time required to construct
  $\mathcal{A}$. We can construct all vertices $(C,R) \in V_C$ in time 
  $\bigoh(|V(G)|^{\omega}|E(G)|)$ by computing the set of all components of
  $G \setminus C$ for every cop-position $C$. Note that within the
  same time we can additionally compute and store the guards
  $\guards(R)$ for every component $R$. 
  For each vertex in $(C,R) \in V_C$ we
  can then compute the associated vertices $(C,C',R) \in V_R$ and add
  the necessary arcs by
  enumerating all sets $C' \subseteq V(G)$ with $\guards(R) \subseteq
  C'$ and checking for each of those whether $R$ and $C'$ satisfy CM3.
  Enumerating the sets $C' \subseteq V(G)$ with $\guards(R) \subseteq
  C'$ can be achieved in time $\bigoh(|V(G)|^\omega)$ (using the fact that
  we stored $\guards(R)$ for every component $R$). Moreover,
  determining whether $R$ and $C'$ satisfy CM3 can be achieved in time
  $\bigoh(\omega|R|)=\bigoh(\omega|V(G)|)$ by going over all vertices $r \in R$ and
  verifying that $r \in C'$ or that it is not smaller than any vertex
  in $C'$. Hence computing all vertices of $\mathcal{A}$ and all arcs
  from vertices in $V_C$ to vertices in $V_R$ can be achieved in time
  at most
  $\bigoh(|V(G)|^{\omega}|E(G)|+|V(G)|^{2\omega+1}\omega)$, which for the
  natural assumption that $\omega >0$ is at most $\bigoh(|V(G)|^{2\omega+2})$.
  Finally, we need to add the arcs between vertices in $(C,C',R) \in V_R$ and
  the vertices $(C',R') \in V_C$. Note that there is an arc from
  $(C,C',R) \in V_R$ to $(C',R') \in V_C$ if and only if $R'\subseteq
  R$ and moreover for every component $R'$ of $G \setminus C'$ either
  it is a subset of $R$ or it is disjoint with $R$, which can be
  checked in constant time. Hence the total time required for this
  last step is equal to the number of vertices in $V_R$ times $|V(G)|$ which is at
  most $\bigoh(|V(G)|^{2\omega+2})$.
  It follows that the time required to construct $\mathcal{A}$ is at most $\bigoh(|V(G)|^{2\omega+2})$.
  Once the arena is constructed the winning regions as well as the winning
  strategies for both players can be computed in time $\bigoh(|V_C\cup
  V_R|+|A|)\in \bigoh(|V(G)|^{2\omega+2})$.
}
\end{proof}

The next theorem summarizes our second algorithm for computing
dependency treewidth. The core distinction here lies in the fact that the running time does not depend on the dependency treewidth, but rather on the poset-width. This means that the algorithm can precisely compute the dependency treewidth even when this is large, and it will perform better than Theorem~\ref{thm:comput} for formulas with ``tighter'' dependency structures (e.g., formulas which utilize the full power of quantifier alternations).

\begin{THE}\label{the:exact-oew}
   There is an algorithm running in time $\bigoh((|V(G)|^{k}k^2)$ that, given a graph $G$
   and a poset $\PPP=(V(G),\leq^\PPP)$ of width $k$ and $\omega\in \Nat$, 
   determines whether $G$ and $\PPP$ admit a dependency elimination ordering of width
   at most $\omega$, and if yes outputs such a dependency elimination ordering.
%
%
\end{THE}

\begin{proof}
  To decide whether $G$ has a dependency elimination ordering 
  of width at most $\omega$, we first build an auxiliary 
  directed graph $H$ as follows.

  The vertex set of $H$ consists of all pairs $(D,d)$ such $D
  \subseteq V(G)$ is a downward closed set and $d \in D$ is a maximal
  element of $D$ such that
  $|N_{G[D\setminus d]}(d)|\leq \omega$. Additionally, $H$ contains the
  vertices $(V(G),\emptyset)$ and $(\emptyset,\emptyset)$.
  Furthermore, there is an arc
  from $(D,d)$ to $(D',d')$ of $H$ if and only if $D'=D \cup \{d'\}$ or
  $D=D'=V(G)$ and $d'=\emptyset$.
  This completes the construction of
  $H$. It is immediate that
  $G$ has a dependency elimination ordering of width at
  most $\omega$ if and only if there is a directed path in $H$ from
  $(\emptyset,\emptyset)$ to $(V(G),\emptyset)$. Hence, given $H$ we
  can decide whether $G$ has a dependency elimination ordering
  of width at most $\omega$ (and output it, if it exists) in time $\bigoh(|V(H)|\log
  (|V(H)|)+E(H))$ (e.g., by using Dijkstra's algorithm). It remains to
  analyze the time required to construct $H$ (as well as its size).

  Let $k$ be the width of the poset $\PPP$.
\sv{By the algorithm of Felsner, Raghavan and Spinrad~\cite{frs03},}\lv{Due to Proposition~\ref{pro:comp-chain-part},} we can compute a chain
  partition $\mathcal{C}=(W_1,\dotsc,W_k)$ of width $k$ in time
  $\bigoh(k\cdot |V(G)|^2)$. Note that every downward closed $D$ set can be
  characterized by the position of the maximal element in $D$ on each
  of the chains $W_1,\dotsc,W_k$, we obtain that there are at most
  $|V(G)|^k$ downward closed sets.
  Hence, $H$ has at most
  $\bigoh(|V(G)|^k(k+1))$ vertices its vertex set can be
  constructed in time $\bigoh(|V(G)|^{k}(k+1))$. Since every vertex
  $(D,d)$ of $H$ has at most $k+1$
  possible out-neighbors, we can construct the arc set of $H$ in time
  $\bigoh(|V(G)|^{k}k^2)$.

  Hence,
  the total time required to construct $H$ is
  $\bigoh((|V(G)|^{k}k^2)$ which dominates the time
  required to find a shortest path in $H$, and so the runtime follows.
\end{proof}

\section{Concluding Notes}
Dependency treewidth is a promising decompositional parameter for QBF which overcomes the key shortcomings of previously introduced structural parameters; its advantages include a single-exponential running time, a refined and flexible approach to variable dependencies, and the ability to compute decompositions. It also admits several natural characterizations that show the robustness of the parameter and allows the computation of resolution proofs.

The presented algorithms for computing dependency elimination orderings leave open the question of whether this problem admits a fixed-parameter algorithm (parameterized by dependency treewidth). We note that the two standard approaches for computing treewidth fail here. In particular, the well-quasi-ordering approach with respect to minors does not work since the set of ordered graphs can be observed not to be well-quasi ordered w.r.t.\ the ordered minor relation~\cite{BonaSpielman00}. On the other hand, the records used in the second approach~\cite{Bodlaender93} do not provide sufficient information in our ordered setting.


\paragraph{Acknowledgments} The authors wish to thank the anonymous reviewers for their helpful comments. Eduard Eiben acknowledges support by the Austrian Science Fund (FWF, projects P26696 and W1255-N23). Robert Ganian is also affiliated with FI MU, Brno, Czech Republic.

\bibliographystyle{plain}
\bibliography{literature}

\begin{thebibliography}{10}

\bibitem{AdlerWeyer12}
Isolde Adler and Mark Weyer.
\newblock Tree-width for first order formulae.
\newblock {\em Logical Methods in Computer Science}, 8(1), 2012.

\bibitem{AtseriasOliva14}
Albert Atserias and Sergi Oliva.
\newblock Bounded-width {QBF} is pspace-complete.
\newblock {\em J. Comput. Syst. Sci.}, 80(7):1415--1429, 2014.

\bibitem{BiereLonsing10}
Armin Biere and Florian Lonsing.
\newblock Integrating dependency schemes in search-based {QBF} solvers.
\newblock In {\em Theory and Applications of Satisfiability Testing - SAT
  2010}, volume 6175 of {\em LNCS}, pages 158--171. Springer, 2010.

\bibitem{Bodlaender93}
Hans~L. Bodlaender.
\newblock A linear time algorithm for finding tree-decompositions of small
  treewidth.
\newblock In {\em Proceedings of the Twenty-Fifth Annual {ACM} Symposium on
  Theory of Computing, May 16-18, 1993, San Diego, CA, {USA}}, pages 226--234,
  1993.

\bibitem{Capellithesis16}
Florent Capelli.
\newblock {\em Structural restrictions of CNF-formulas: applications to model
  counting and knowledge compilation}.
\newblock PhD thesis, Universit\'e Paris Diderot, 2016.

\bibitem{ChenDalmau16}
Hubie Chen and V{\'{\i}}ctor Dalmau.
\newblock Decomposing quantified conjunctive (or disjunctive) formulas.
\newblock {\em {SIAM} J. Comput.}, 45(6):2066--2086, 2016.

\bibitem{DavisPutnam60}
M.~Davis and H.~Putnam.
\newblock A computing procedure for quantification theory.
\newblock {\em J. of the ACM}, 7(3):201--215, 1960.

\bibitem{Diestel95}
Reinhard Diestel.
\newblock Graph minors 1: {A} short proof of the path-width theorem.
\newblock {\em Combinatorics, Probability {\&} Computing}, 4:27--30, 1995.

\bibitem{Diestel12}
Reinhard Diestel.
\newblock {\em Graph Theory, 4th Edition}, volume 173 of {\em Graduate texts in
  mathematics}.
\newblock Springer, 2012.

\bibitem{DowneyFellows13}
Rodney~G. Downey and Michael~R. Fellows.
\newblock {\em Fundamentals of Parameterized Complexity}.
\newblock Texts in Computer Science. Springer Verlag, 2013.

\bibitem{EglyEiterTW00}
Uwe Egly, Thomas Eiter, Hans Tompits, and Stefan Woltran.
\newblock Solving advanced reasoning tasks using quantified boolean formulas.
\newblock In {\em Proceedings of the Seventeenth National Conference on
  Artificial Intelligence and Twelfth Conference on on Innovative Applications
  of Artificial Intelligence, July 30 - August 3, 2000, Austin, Texas, {USA.}},
  pages 417--422, 2000.

\bibitem{EibenGanianOrdyniak16}
Eduard Eiben, Robert Ganian, and Sebastian Ordyniak.
\newblock Using decomposition-parameters for {QBF:} mind the prefix!
\newblock In {\em Proceedings of the Thirtieth {AAAI} Conference on Artificial
  Intelligence, February 12-17, 2016, Phoenix, Arizona, {USA.}}, pages
  964--970, 2016.

\bibitem{frs03}
S.~Felsner, V.~Raghavan, and J.~Spinrad.
\newblock Recognition algorithms for orders of small width and graphs of small
  dilworth number.
\newblock {\em Order}, 20(4):351--364, 2003.

\bibitem{FlumGrohe06}
J\"{o}rg Flum and Martin Grohe.
\newblock {\em Parameterized Complexity Theory}, volume XIV of {\em Texts in
  Theoretical Computer Science. An EATCS Series}.
\newblock Springer Verlag, Berlin, 2006.

\bibitem{GraedelThomasWilke02}
Erich Gr{\"{a}}del, Wolfgang Thomas, and Thomas Wilke, editors.
\newblock {\em Automata, Logics, and Infinite Games: {A} Guide to Current
  Research [outcome of a Dagstuhl seminar, February 2001]}, volume 2500 of {\em
  Lecture Notes in Computer Science}. Springer, 2002.

\bibitem{KleinebuningBubeck09}
Hans Kleine~B{\"u}ning and Uwe Bubeck.
\newblock Theory of quantified boolean formulas.
\newblock In A.~Biere, M.~J.~H. Heule, H.~van Maaren, and T.~Walsh, editors,
  {\em Handbook of Satisfiability}, volume 185 of {\em Frontiers in Artificial
  Intelligence and Applications}, chapter~23, pages 735--760. IOS Press, 2009.

\bibitem{KleineBuningLettman99}
Hans Kleine~B{\"u}ning and Theodor Lettman.
\newblock {\em Propositional logic: deduction and algorithms}.
\newblock Cambridge University Press, Cambridge, 1999.

\bibitem{Kloks94}
T.~Kloks.
\newblock {\em Treewidth: Computations and Approximations}.
\newblock Springer Verlag, Berlin, 1994.

\bibitem{OtwellRemshagenTruemper04}
Charles Otwell, Anja Remshagen, and Klaus Truemper.
\newblock An effective {QBF} solver for planning problems.
\newblock In {\em Proceedings of the International Conference on Modeling,
  Simulation {\&} Visualization Methods, {MSV} '04 {\&} Proceedings of the
  International Conference on Algorithmic Mathematics {\&} Computer Science,
  {AMCS} '04, June 21-24, 2004, Las Vegas, Nevada, {USA}}, pages 311--316.
  CSREA Press, 2004.

\bibitem{Papadimitriou94}
Christos~H. Papadimitriou.
\newblock {\em Computational Complexity}.
\newblock Addison-Wesley, 1994.

\bibitem{PulinaT09}
Luca Pulina and Armando Tacchella.
\newblock A structural approach to reasoning with quantified boolean formulas.
\newblock In Craig Boutilier, editor, {\em {IJCAI} 2009, Proceedings of the
  21st International Joint Conference on Artificial Intelligence, Pasadena,
  California, USA, July 11-17, 2009}, pages 596--602, 2009.

\bibitem{PulinaT10}
Luca Pulina and Armando Tacchella.
\newblock An empirical study of {QBF} encodings: from treewidth estimation to
  useful preprocessing.
\newblock {\em Fundam. Inform.}, 102(3-4):391--427, 2010.

\bibitem{Rintanen99}
J.~Rintanen.
\newblock Constructing conditional plans by a theorem-prover.
\newblock {\em J. Artif. Intell. Res.}, 10:323--352, 1999.

\bibitem{SabharwalAGHS06}
Ashish Sabharwal, Carlos Ans{\'{o}}tegui, Carla~P. Gomes, Justin~W. Hart, and
  Bart Selman.
\newblock {QBF} modeling: Exploiting player symmetry for simplicity and
  efficiency.
\newblock In {\em Theory and Applications of Satisfiability Testing - {SAT}
  2006, 9th International Conference, Seattle, WA, USA, August 12-15, 2006,
  Proceedings}, pages 382--395, 2006.

\bibitem{SamerSzeider09a}
Marko Samer and Stefan Szeider.
\newblock Backdoor sets of quantified {B}oolean formulas.
\newblock {\em Journal of Autom. Reasoning}, 42(1):77--97, 2009.

\bibitem{SlivovskySzeider12}
Friedrich Slivovsky and Stefan Szeider.
\newblock Computing resolution-path dependencies in linear time.
\newblock In Alessandro Cimatti and Roberto Sebastiani, editors, {\em Theory
  and Applications of Satisfiability Testing - SAT 2012}, volume 7317 of {\em
  Lecture Notes in Computer Science}, pages 58--71. Springer Verlag, 2012.

\bibitem{SlivovskySzeider14}
Friedrich Slivovsky and Stefan Szeider.
\newblock Variable dependencies and q-resolution.
\newblock In {\em Theory and Applications of Satisfiability Testing - {SAT}
  2014 - 17th International Conference, Held as Part of the Vienna Summer of
  Logic, {VSL} 2014, Vienna, Austria, July 14-17, 2014. Proceedings}, volume
  8561 of {\em Lecture Notes in Computer Science}, pages 269--284. Springer,
  2014.

\bibitem{SlivovskySzeider16}
Friedrich Slivovsky and Stefan Szeider.
\newblock Quantifier reordering for {QBF}.
\newblock {\em J. Autom. Reasoning}, 56(4):459--477, 2016.

\bibitem{BonaSpielman00}
Daniel~A Spielman and Mikl{\'o}s B{\'o}na.
\newblock An infinite antichain of permutations.
\newblock {\em Electron. J. Combin}, 7:N2, 2000.

\bibitem{StockmeyerMeyer73}
L.~J. Stockmeyer and A.~R. Meyer.
\newblock Word problems requiring exponential time.
\newblock In {\em Proceedings of the 5th Annual {ACM} Symposium on Theory of
  Computing, April 30 - May 2, 1973, Austin, Texas, {USA}}, pages 1--9. ACM,
  1973.

\bibitem{Szeider04b}
Stefan Szeider.
\newblock On fixed-parameter tractable parameterizations of {S}{A}{T}.
\newblock In Enrico Giunchiglia and Armando Tacchella, editors, {\em Theory and
  Applications of Satisfiability, 6th International Conference, SAT 2003,
  Selected and Revised Papers}, volume 2919 of {\em Lecture Notes in Computer
  Science}, pages 188--202. Springer Verlag, 2004.

\bibitem{VanGelder11}
Allen Van~Gelder.
\newblock Variable independence and resolution paths for quantified boolean
  formulas.
\newblock In Jimmy Lee, editor, {\em Principles and Practice of Constraint
  Programming - CP 2011}, volume 6876 of {\em Lecture Notes in Computer
  Science}, pages 789--803. Springer Verlag, 2011.

\end{thebibliography}

\end{document}